\documentclass{article}

\usepackage{microtype}
\usepackage{mathtools}
\usepackage{amsmath}
\usepackage{amsthm}
\usepackage{xcolor}
\usepackage{algorithm}
\usepackage{algpseudocode}

\DeclareMathOperator{\poly}{poly}

\newtheorem{defn}{Definition}
\newtheorem{ex}{Example}
\newtheorem{obs}{Observation}
\newtheorem{thm}{Theorem}
\newtheorem{lem}[thm]{Lemma}
\newtheorem{prop}[thm]{Proposition}
\newtheorem{hyp}[thm]{Hypothesis}
\newtheorem{fact}{Fact}

\newcommand{\NP}{\textsc{NP}}
\newcommand{\AND}{\ensuremath{\wedge}}
\newcommand{\OR}{\ensuremath{\vee}}
\newcommand{\NOT}{\ensuremath{\neg}}
\newcommand{\CM}{\textsc{CM}}
\newcommand{\ACM}{\textsc{ACM}}
\newcommand{\HAZARD}{\textsc{Hazard}}

\newcommand{\DNFFALSE}{\textsc{Dnffalse}}
\newcommand{\CNFSAT}{\textsc{Cnfsat}}
\newcommand{\DNF}{\textsc{Dnf}}
\newcommand{\CNF}{\textsc{Cnf}}
\newcommand{\SETH}{\textsc{Seth}}

\title{On the complexity of detecting
hazards\footnote{https://doi.org/10.1016/j.ipl.2020.105980}}

\author{
  Balagopal Komarath\\
  Saarland University\\
  Germany\\
  \texttt{bkomarath@rbgo.in}
  \and
  Nitin Saurabh\footnote{This work was done when the author was affiliated with the Max Planck Institute for Informatics, Saarland Informatics Campus, Saarbr\"{u}cken, Germany.}\\
  Technion - IIT\\
  Israel\\
  \texttt{nitinsau@cs.technion.ac.il}}

\begin{document}

\maketitle

\begin{abstract}

  Detecting and eliminating logic hazards in Boolean circuits is
  a fundamental problem in logic circuit design. We show that
  there is no $O(3^{(1-\epsilon)n} \poly(s))$ time algorithm, for
  any $\epsilon > 0$, that detects logic hazards in Boolean
  circuits of size $s$ on $n$ variables under the assumption that
  the strong exponential time hypothesis is true.  This lower
  bound holds even when the input circuits are restricted to be 
  formulas of depth four. We also present a polynomial time
  algorithm for detecting $1$-hazards in \DNF\
  (or, $0$-hazards in \CNF) formulas. Since $0$-hazards in \DNF\
  (or, $1$-hazards in \CNF) formulas are easy to eliminate, this algorithm can
  be used to detect whether a given \DNF\ or \CNF\ formula has a
  hazard in practice.

\end{abstract}

\section{Introduction}
\label{sec:intro}

In logic design, one typically extends the Boolean domain $\{0,
1\}$ with a third value denoted by `$u$' to indicate the presence
of unstable voltage levels. In other words, it indicates that the value
of a Boolean variable is \emph{unknown}. Useful computation can be performed
even in the presence of unstable/unknown values.  For example, consider a
Boolean circuit over $\AND$, $\OR$, and $\NOT$ gates
that takes $2n$ bits as input and decides whether a
majority of the input bits are set to $1$.
Clearly if at least $n+1$ of its inputs are $1$
(or, $0$) and the rest of the inputs are $u$, the circuit should
ideally output $1$ (resp.,~$0$).\footnote{See Table~\ref{tab:truth} for the definition of Boolean gates in the presence of $u$.}
However, if the circuit outputs $u$ on
such inputs, then the circuit is said to have a \emph{hazard}.
A priori the circuit
may be hazard-free or not, but if it is monotone
(only $\AND$ and $\OR$ gates are allowed)
then it must be hazard-free~\cite{IKLLMS18}. 

It is well-known that popular physical realizations of the basic
logic gates are hazard-free. For example, if we feed a $0$ and a
$u$ to an AND ($\AND$) gate, then the $\AND$ gate will output a
$0$ (Table~\ref{tab:truth}). 
  However, it is not
necessary that a circuit constructed from hazard-free logic gates
is hazard-free. For example, the smallest circuit implementing a
one bit multiplexer has a hazard (see, e.g., \cite{IKLLMS18}).

For a logic circuit designer, it is desirable that every circuit
they construct is hazard-free. In a recent paper, Ikenmeyer
\textit{et al}.~\cite{IKLLMS18} showed that there are
$n$-variable Boolean functions with polynomial (in $n$) size
circuits such that any hazard-free circuit implementation for the
same function must use exponentially many gates. Therefore,
constructing hazard-free circuits is not always feasible. They
also showed that even the computational problem of detecting
whether a circuit has a hazard is $\NP$-complete.

Eichelberger's algorithm \cite{EIC65} for detecting hazards in a
circuit enumerates all minterms and maxterms of the Boolean
function computed by the circuit and evaluates the circuit on
each of them.  Since an $n$-variable Boolean function can have as
many as $\Omega(3^n/n)$ minterms \cite{CM78}, this algorithm is
not always efficient.  Since this problem is $\NP$-complete, one
cannot hope to obtain a polynomial time algorithm that works in
general. In such cases, moderately exponential time algorithms
are sought over algorithms that employ brute-force search.  For
example, it is known \cite{XN17} that the independent set problem
has a $O(\poly(n)1.19^n)$ time algorithm that performs much
better than the brute-force $O(\poly(n)2^n)$ time algorithm. Is
it possible to obtain such a moderately exponential time
algorithm for hazard detection in circuits?

In this work, we show that the $O(3^n\poly(s))$ time algorithm
for hazard detection on input circuits of size $s$ over $n$
variables  is almost optimal under a widely held conjecture known
as the \emph{strong exponential time hypothesis} (\SETH).  \SETH\
implies that there is no $O(2^{(1-\varepsilon)n}\poly(m))$ time
algorithm, for any $\varepsilon > 0$, for checking whether an
$n$-variable, $m$-clause \CNF\ is satisfiable. We show that there
is no $O(3^{(1-\varepsilon)n}\poly(s))$ time algorithm, for any
$\varepsilon >0$, for hazard detection on circuits of size $s$
over $n$ variables assuming \SETH. In fact, we show that this is
true even when the input circuits are restricted to be formulas
of depth four.

We also give a polynomial time algorithm to detect whether a
given \DNF\ formula has a 1-hazard. Since 0-hazards in \DNF\
formulas are easy to eliminate, this algorithm can be used to
check whether a given \DNF\ formula has a hazard in practice.  We
remark that, using duality of hazards, this also implies a hazard
detection algorithm for \CNF\ formulas. 

\section{Preliminaries}
\label{sec:prelim}

We study Boolean functions $f : \{0,1\}^n \to \{0,1\}$ on $n$
variables where $n$ is an arbitrary natural number. We are
interested in Boolean circuits over AND ($\AND$), OR ($\OR$), and
NOT ($\NOT$) gates computing such functions. We recall Boolean
circuits are directed acyclic graphs with a unique sink node
(output gate), where the source nodes (input gates) are labeled
by literals, i.e., $x_i$ or $\neg x_i$ for $i \in [n]$ and
non-source nodes are labeled by $\wedge$ or $\vee$ gates.  The
\emph{depth} of a gate in the circuit is defined as the maximum
number of $\AND$ or $\OR$ gates occurring on any path from an
input gate to this gate (inclusive). (Note that $\NOT$ gates do
not contribute to depth.) The \emph{depth} of a circuit is then
defined to be the depth of the output gate. In particular, $\CNF$
and $\DNF$ formulas have depth two.  We recall \emph{formulas}
are circuits such that the underlying undirected graph is a tree,
i.e., every gate other than the output gate has out-degree
exactly $1$.  

We refer to constant depth formulas where all gates of the same depth
are of the same type
by the sequence of $\AND$ and $\OR$ starting from the output gate. For
example, $\CNF$ formulas are $\AND\OR$ formulas. 

In our setting the input variables to circuits are allowed to take
an unstable value, denoted by $u$, in addition to the usual stable
values $0$ and $1$. The truth tables for gates in the basis
$\{\wedge, \vee, \neg\}$ in the presence of unstable values are given in
Table~\ref{tab:truth}. The truth table for larger fan-in $\AND$ and $\OR$
gates in the presence of $u$ can be similarly defined using associativity.
Thus, we can evaluate circuits on inputs from $\{u, 0, 1\}^n$
in the usual inductive fashion.
\begin{table}[ht]
  \centering
  \begin{tabular}{c|ccc}

    $\AND$ & $u$ & 1 & 0 \\ \hline
    $u$      & $u$ & $u$ & 0 \\
    1      & $u$ & 1 & 0  \\
    0      & 0 & 0 & 0 \\

  \end{tabular}
  \quad
  \begin{tabular}{c|ccc}

    $\OR$ & $u$ & 0 & 1 \\ \hline
    $u$     & $u$ & $u$ & 1 \\
    0     & $u$ & 0 & 1 \\
    1     & 1 & 1 & 1

  \end{tabular}
  \quad
  \begin{tabular}{c|c}

    $\NOT$ &   \\ \hline
    $u$      & $u$ \\
    0      & 1 \\
    1      & 0 

  \end{tabular}
  \caption{Truth table for AND, OR, and NOT gates.}
  \label{tab:truth}
\end{table}

We now formally introduce the notion of \emph{hazard}. 

\begin{defn}
  A string $b\in \{0, 1\}^n$ is called a \emph{resolution} of a
  string $a\in \{u, 0, 1\}^n$ if $b$ can be obtained from $a$ by
  only changing the unstable values in $a$ to stable values.
\end{defn}
For example, strings $0100$, $0110$, $1110$, and $1100$ are all possible resolutions of the string $u1u0$, but $0111$ is not.
\begin{defn}
  A circuit $C$ implementing a Boolean function has a
  \emph{$1$-hazard} (or,~\emph{$0$-hazard}) on an input $a\in \{u, 0, 1\}^n$
  if and only if $C(a) = u$ yet for all resolutions $b$ of
  $a$, the value $C(b)$ is $1$ (resp.,~$0$). A circuit has a
  \emph{hazard} if it has a $1$-hazard or a $0$-hazard.
\end{defn}
\begin{ex}
  \label{ex:hazards}
  Consider the $\DNF$ formula $F = (x_1 \wedge x_2) \vee (\neg
  x_1 \wedge x_2) \vee (\neg x_1 \wedge \neg x_2)$ implementing
  the function $f$ that evaluates to $0$ only when $x_1 = 1$ and
  $x_2 = 0$. Consider the input $x_1x_2 = 0u$. The function $f$
  evaluates to $1$ on both resolutions of $0u$. But, the formula
  $F$ evaluates to $u$ on input $x_1 = 0$, $x_2 = u$. Therefore,
  $F$ has a $1$-hazard at the input $0u$. We note that $u1$ is another
  input where $F$ has a hazard.
\end{ex}
We remark that being hazard-free or not is a property of the formula or circuit and not a property of the function being computed by them. 

In this paper, we are interested in the time complexity of the
following language.
\begin{defn}
  The language \HAZARD\ consists of all circuits that have
  hazards.
\end{defn}
The hazards in a circuit implementing a function $f$ are closely
related to the minterms and maxterms of $f$. A definition of
these concepts and their relationship to hazards in circuits
follows.

\begin{defn}
  \label{defn:implicants}

    A \emph{$1$-implicant} (\emph{$0$-implicant}) of a Boolean
    function $f$ on variables $x_1, \dotsc, x_n$ is an AND
    (resp.,~OR) over a subset $I$ of literals $x_1,\dotsc,
    x_n$,$\neg{x}_1, \ldots , \neg{x}_n$ such that for any
    assignment $a\in\{0, 1\}^n$, if $I(a) = 1$ (resp.,~$I(a) =
    0$),\footnote{By $I(a)$, we mean the AND (or, OR) function
    over the set $I$ of literals evaluated at $a$. We often
    overload notation to denote both the set of literals in an
    implicant and the function by the same notation $I$.} then
    $f(a) = 1$ (resp.,~$f(a) = 0$).  In such a case the
    assignment $a$ is said to be \emph{covered} by the implicant
    $I$. The \emph{size} of an implicant is defined to be the
    size of the set $I$.

\end{defn}

\begin{ex}
  \label{ex:implicants}
  Consider the function $f$ from Example~\ref{ex:hazards}.
  The only $0$-implicant of $f$ is $\neg x_1 \vee x_2$ and it is of size $2$.
  The function has five $1$-implicants $x_1 \wedge x_2$,
  $\neg x_1 \wedge x_2$, $\neg x_1 \wedge \neg x_2$, $x_2$, and
  $\neg x_1$. The assignments $x_1 = 0$, $x_2 = 0$ and $x_1 = 0$,
  $x_2 = 1$ are covered by the $1$-implicant $\neg x_1$.
\end{ex}

\begin{defn}
  A $1$-implicant ($0$-implicant) that is minimal with respect to set
  containment is called a \emph{minterm} (resp.,~\emph{maxterm}).
\end{defn}

\begin{ex}
  \label{ex:minterm-maxterm}
  Continuing from Example~\ref{ex:implicants}, we note that $f$ 
  has one maxterm $\neg x_1 \vee x_2$ and two minterms,
  namely $\neg x_1$ and $x_2$.
\end{ex}

We have the following well-known \emph{cross-intersection} property of
the set of all minterms and the set of all maxterms. 

\begin{fact}
  \label{fact:min-max-intersection}
  Let $S$ be any minterm and $T$ be any maxterm for a function $f$.
  Then, $S \cap T \neq \emptyset$.  
\end{fact}
\begin{proof}
  Suppose not, then there exists an assignment $a$ such that $S(a) =1$
  and $T(a) = 0$. But then from Definition~\ref{defn:implicants}
  we have $f(a) = 1$ as well as $f(a) = 0$, which is a contradiction. \qed
\end{proof}

An implicant can be naturally represented as an assignment of
variables to $\{u,0,1\}$ where the variables not in the
implicant are set to $u$ and the ones present are set to $0$ or
$1$ so as to make the corresponding literal evaluate to $1$ for
$1$-implicants (or, $0$ for $0$-implicants). By evaluating a
circuit at a minterm or maxterm, we mean evaluating the circuit
on the corresponding assignment in $\{u,0,1\}^n$.

\begin{obs}
A circuit $C$ implementing a Boolean function $f$ has a
$1$-hazard ($0$-hazard) if and only if it has hazard at a
minterm (resp.,~maxterm).
\end{obs}
\begin{proof}
  Given an input $a$ at which $C$ has hazard,
  consider a minterm or maxterm that covers $a$.  It is easily seen that 
  the output of evaluating $C$ on this minterm or maxterm is $u$, because
  changing stable values in the input to $u$ cannot cause the
  output to go from $u$ to a stable value. \qed
\end{proof}

Since any minterm or maxterm of an $n$-variable Boolean function can be
represented by a string from $\{u,0,1\}^n$, there can be at
most $3^n$ minterms for a Boolean function. How tight is this
upper bound? Chandra and Markowsky \cite{CM78} gave an improved upper bound
of $O(3^n/\sqrt{n})$ and also gave an example to show that this upper bound
is almost tight. We recall the function witnessing this lower bound now.
We call it the Chandra-Markowsky (\CM) function.

The Chandra-Markowsky (\CM) function
\cite{CM78} on $N = 3n$ variables for any natural $n$ is defined
as follows: it evaluates to $1$ if and only if at least $n$ of
the variables are set to $1$ and at least $n$ of the variables
are set to $0$. This function has $\binom{3n}{n}\binom{2n}{n}
=\Theta(3^N/N)$ minterms.  Therefore, it has
almost the maximum possible number of minterms.

The Strong Exponential Time Hypothesis (\SETH) is a conjecture
introduced by Impagliazzo, Paturi and Zane \cite{IP01, IPZ01} to
address the time complexity of the \CNF\ satisfiability problem
(\CNFSAT).  It has been used to establish conditional lower
bounds for many $\NP$-complete problems (e.g., \cite{LMS11,
Cygan12}) and problems with polynomial time algorithms (e.g.,
\cite{Bringmann15, Abboud15, VWilliams18}).

\begin{hyp}[\SETH\ \cite{IP01, IPZ01}]
  For every $\varepsilon >0$, there exists an integer $k \geq 3$
  such that no algorithm can solve $k$-\CNFSAT \footnote{Every clause in the \CNF\ is defined on at most $k$ literals.} on $n$
  variables in $O(2^{(1-\varepsilon)n})$ time. 
\end{hyp}

To establish our lower bound, we reduce from the \DNF\
falsifiability problem ($\DNFFALSE$): \emph{Given a \DNF\ formula
as an input, determine whether there exists an assignment that
falsifies it.}

This problem clearly has the same time complexity as the
$\CNFSAT$ problem. \SETH\ implies that there is no
$O(2^{(1-\epsilon)n} \poly(s))$ time algorithm ,for any $\epsilon
> 0$, for \DNFFALSE\ where $n$ is the number of variables and $s$
is the number of clauses.

\section{A tight lower bound}

The idea behind the proof is as follows: The given \DNF\ formula
$F$ on $n$ variables has $2^n$ assignments. We construct a
formula $F'$ on $m \sim \log_3(2) n$ variables that implements a
function with more than $2^n$ minterms. This allows us to map
each assignment of variables in $F$ to a distinct minterm of
the function implemented by $F'$. We then show that $F$ is
falsifiable by an assignment $a$ if and only if the formula $F'$
has a hazard at the minterm $b$ that corresponds to $a$ in the
mapping. First, we define the function that is going to be
implemented by $F'$ and prove some important properties related
to it.

Let $s$ be any natural number that is a multiple of $3$.  We now
define the auxiliary function $\ACM$ that will be used in our
reduction. It is defined on $sn$ variables which are partitioned
into $n$ groups of $s$ variables each.  We simply compose the
\textsc{AND} function on $n$ variables with the $\CM$ function on
$s$ variables to define the $\ACM$ function.  That is,
\begin{align} \label{eq:ACM}
\ACM(X_1, \dotsc, X_n) = \CM(X_1) \AND \dotsm \AND \CM(X_n)
\end{align}
where $X_i$ denotes the $i$-th group of $s$ variables.  We denote
the $\CM$ function on the $i$-th group of variables $X_i$ by
$\CM_i$.

The following proposition characterizes the minterms and maxterms
of the $\ACM$ function.

\begin{prop}
  \label{prop:minterm-maxterm-ACM}

  The following statements are true:

  \begin{enumerate}

  \item[(i)] The set of minterms of $\ACM$ is the direct product
  of the set of minterms of the $n$ disjoint $s$-variable $\CM$
  functions.

  \item[(ii)] The set of maxterms of $\ACM$ is given by the union
  of the set of maxterms of $\CM_i$ for $1\leq i \leq n$. 

  \end{enumerate}

\end{prop}

\begin{proof}

  \phantom{We show}

  \begin{enumerate}

  \item[\textit{(i)}] We show a one-to-one correspondence between
  the set of minterms of $\ACM$ and the direct product of the set
  of minterms of $\CM_j$ for $j \in [n]$.  For a minterm $I$ of
  $\ACM$, let $I_j$ be the restriction of $I$ to the variables in
  $X_j$ for each $j$. Since $\ACM$ evaluates to $1$ on $I$,
  $\CM_j$ must evaluate to $1$ on $I_j$.  Thus, $I_j$ is a
  $1$-implicant of $\CM_j$ for each $j$. We now argue that in
  fact it is a minterm. Suppose not, then there exists a $j$ such
  that $I_j$ is not a minterm of $\CM_j$. However it must contain
  a minterm, since it is a $1$-implicant. Let $I'_j \subset I_j$
  be the minterm contained in $I_j$. By replacing the part of
  $I_j$ in $I$ by $I'_j$ we obtain $I'$. Clearly, $I' \subset I$
  is a $1$-implicant of $\ACM$. Thus, we have a contradiction to
  the fact that $I$ is a minterm.
    
  On the other hand, given minterms $I_j$ of $\CM_j$ for each $j
  \in [n]$, their union $I = \cup_{j \in [n]}I_j$ is a minterm of
  $\ACM$. Suppose not, then there exists $I' \subset I$ that is a
  minterm. Since $I' \subset I$, then there exists $j \in [n]$
  such that $I'_j \subset I_j$. Thus, we obtain a contradiction
  to $I_j$ being a minterm of $\CM_j$.

  \item[\textit{(ii)}] For some $j \in [n]$, let $I$ be a maxterm
  of $\CM_j$. Then it is easily seen that $I$ is also a maxterm
  of $\ACM$.  To prove the other direction, for a maxterm $I$ of
  $\ACM$, we argue that there exists a unique $j \in [n]$ such
  that $I$ is a maxterm of $\CM_j$.  Clearly, there exists a $j
  \in [n]$ such that $I$ is a $0$-implicant of $\CM_j$. Since $I$
  is a maxterm of $\ACM$, it must only set variables in $X_j$.
  And therefore, the term $I$ is a $0$-implicant of the unique
  $\CM_j$.  Hence, the term $I$ must also be a maxterm of
  $\CM_j$.

  \end{enumerate}\qed
  
\end{proof}

Huffman \cite{Huf:57} showed that for any Boolean function $f$
with the set $\mathcal{M}$ of all minterms, the \DNF\ $F =
\bigvee_{I\in \mathcal{M}}I$ is hazard-free.  For example,
consider the function $f$ defined in Example~\ref{ex:hazards}.
From Example~\ref{ex:minterm-maxterm} we know that $\neg x_1$ and $x_2$ are
the only two minterms of it. Therefore, $(\neg x_1) \vee (x_2)$
is the hazard-free DNF implementation for $f$ given by Huffman's construction.

In our reduction,
it will be crucial for us to be able to introduce hazards to the
implementation at specific minterms.  For this purpose we modify
Huffman's hazard-free \DNF\ construction as follows.

\begin{prop}
  \label{prop:hazard-in-f}
  Let $f$ be a function on $n$ variables and
  $S$ be a set of minterms of $f$ where 
  each minterm in $S$ is of size at most $n-1$.
  Then, we can construct a \DNF\ for $f$ that has
  hazards exactly at the minterms in $S$.
\end{prop}
\begin{proof}
  Let $F$ be the hazard-free \DNF\ for $f$ given by Huffman's
  construction. Let $I$ be a minterm in $S$ and $x$ be a variable not
  in $I$. Such a variable exists by assumption.
  Consider the formula $F'$ obtained by replacing 
  the term $I$ in $F$ with two new terms,
  namely $I \AND x$ and $I \AND \bar{x}$. 
  $F'$ computes the same function and has a $1$-hazard at the minterm
  $I$, since the two new terms evaluate to $u$ on $I$ and every
  other term will evaluate to $0$ or $u$ on $I$. For any minterm
  not in $S$, $F'$ evaluates to $1$. Therefore, these
  are the only $1$-hazards. Also, $F'$ has no
  $0$-hazards, because every maxterm and minterm intersects contradictorily.

  We repeat the aforementioned transformation for every minterm in $S$
  to obtain the required \DNF\ for $f$ that has hazards 
  at the minterms in $S$. \qed
\end{proof}

To illustrate we consider our running example, the function $f$ from
Example~\ref{ex:hazards}. We know that $(\neg x_1) \vee (x_2)$
is a hazard-free DNF of $f$. Suppose we want to selectively introduce hazard
only at the minterm $\neg x_1$. Following Proposition~\ref{prop:hazard-in-f},
we modify the hazard-free representation to obtain the following:
\[(\neg x_1 \wedge x_2) \vee (\neg x_1 \wedge \neg x_2) \vee (x_2).\]
Suppose we further wanted to introduce hazard at the minterm $x_2$.
Then, again following Proposition~\ref{prop:hazard-in-f}, we obtain
\[(\neg x_1 \wedge x_2) \vee (\neg x_1 \wedge \neg x_2) \vee (x_2 \wedge x_1),\]
which has hazards at both the minterms. We note that this is
the DNF implementation from Example~\ref{ex:hazards}.

The following lemma applies the above construction to the $\ACM$
function to efficiently introduce hazards in a selective manner. Notice that
any minterm of $\ACM$ has some variable that is not in the
minterm.

\begin{lem}
  \label{lem:hazard-in-ACM}
  Consider the $\ACM$ function on $sn$ variables where $s$ is
  regarded as a constant. For $j \in [n]$, let $\mathcal{M}_j$ be
  the set of all minterms of $\CM_j$.  Further, let $\mathcal{S}
  \subseteq \mathcal{M}_i$ for some $i$.  Then, there is a
  poly-time algorithm that constructs an $\AND\OR\AND$ formula
  for $\ACM$ that has hazards exactly at minterms in the set
  $\mathcal{M}_1 \times \dotsm \times \mathcal{M}_{i-1} \times
  \mathcal{S} \times \mathcal{M}_{i+1}\times \dotsm \times
  \mathcal{M}_n$.
\end{lem}
\begin{proof}
  Let $F_j$ be the hazard-free \DNF\ formula for $\CM_j$ and
  $F'_i$ be the \DNF\ formula for $\CM_i$ that has hazards only
  at minterms in the set $\mathcal{S}$ obtained by
  Proposition~\ref{prop:hazard-in-f}. We output the $\AND\OR\AND$
  formula \((\AND_{j \neq i} F_j) \AND F'_i \) for $\ACM$. The
  size of the formula is $O(n)$ because $s$ is a constant.

  We now argue that this formula has hazards only at minterms in
  the set $\mathcal{M}_1 \times \dotsm \times \mathcal{M}_{i-1}
  \times \mathcal{S} \times \mathcal{M}_{i+1}\times \dotsm \times
  \mathcal{M}_n$.  Since the individual implementation of
  $\CM_j$'s have no $0$-hazards, by
  Proposition~\ref{prop:minterm-maxterm-ACM}~$(ii)$, the formula
  for $\ACM$ has no $0$-hazards.  Now suppose $I$ is a minterm of
  $\ACM$ such that the $\AND\OR\AND$ formula has a hazard at it.
  By Proposition~\ref{prop:minterm-maxterm-ACM}~$(i)$, we know
  that $I_j$ is a minterm of $\CM_j$ for each $j$.  Therefore,
  there exists a $j$ such that the \DNF\ implementation of
  $\CM_j$ has a hazard at $I_j$. But then by construction
  it must be that $j = i$ and $I_j \in \mathcal{S}$. \qed
\end{proof}

We now prove our main theorem.

\begin{thm}
  \label{thm:main}
  If \SETH\ is true, then for any $\epsilon > 0$, there is no
  algorithm for \HAZARD\ that runs in time
  $O({3}^{(1-\epsilon)n}\poly(s))$, even when the inputs are
  formulas of depth four. Here $n$ is the number of variables in
  the formula and $s$ is the size of the formula.
\end{thm}
\begin{proof}
  Let $r$ be a positive integer and $s = s(r)$ be the minimum
  integer  such that $2^r \leq \binom{s}{s/3}\binom{2s/3}{s/3}$.
  We will reduce \DNFFALSE\
  instances on $rn$ variables to instances of \HAZARD\ on $sn$
  variables. In addition, the circuit we output will be an
  $\OR \AND \OR \AND$ formula.  Recall, by our choice, $s$ is a
  multiple of $3$.  For any $\epsilon > 0$, we claim that there
  exists a $\delta > 0$ such that $3^{(1-\epsilon)s} <
  2^{(1-\delta)r}$ for sufficiently large $r$. Let $f(s)$ be the
  number of minterms in the $s$-variable \CM\ function. Then
  $f(s+3) / f(s) \to 27$ as $s\to \infty$. As increasing the
  number of variables by $3$ multiplies the number of assignments
  by $8$, we have $s(r)/r \to \log_{27}(8) = \log_3(2)$ as $r\to
  \infty$.  The claim follows.

  Let $F$ be the input \DNF\ on $rn$ variables. We consider the
  variables of $F$ to be partitioned into $n$ groups $Y_j$, $j
  \in [n]$, of $r$ variables each.  We arbitrarily associate with
  every assignment $\alpha \in \{0,1\}^r$ to the variables in
  $Y_j$ a unique minterm $I_\alpha$ of $\CM_j$ and call this
  bijection $\beta_j$. Recall that $s(r)$ is defined such that
  the number of minterms of $\CM_j$ is at least $2^r$. The
  mapping $\beta_j$ is constant-sized and can be computed easily
  given $j$. The reduction is given in Algorithm~\ref{alg:main}.
  It is easy to see that the algorithm runs in polynomial time and produces
  formulas of depth four. We now argue the correctness of the
  reduction.

  \begin{algorithm}
    \caption{Reduction from \DNFFALSE\ to \HAZARD}
    \label{alg:main}
    \begin{algorithmic}
        \State $F' \leftarrow F$

        \ForAll {literals $\ell$ occurring in $F$}

          \State Replace $\ell$ in $F'$ with LITERAL($\ell$)

        \EndFor

        \State Collapse $\AND$ gates at depth three and four in
        $F'$ to a single layer.

        \State \textbf{return} $F'$

        \Procedure {LITERAL}{$\ell$}
        
        \State Let $\ell$ be $x_j$ or $\neg x_j$.

        \State $i \leftarrow \lceil j/r \rceil$; then, $x_j$ belongs to
        the group $Y_i$.

        \State \( T := \{ \alpha \in \{0,1\}^{|Y_i|} \mid \text{ the literal }\ell\text{ is falsified  by } \alpha\} \)


        \State Recall $\beta_i$ is the bijection from $\{0,1\}^{Y_i}$ to the set of minterms $\mathcal{M}_i$ of $\CM_i$.

        \State $S \leftarrow \beta_i(T)$

        \State $G \leftarrow \AND\OR\AND \text{ formula for }\ACM \text{ given by Lemma~\ref{lem:hazard-in-ACM} on input } S\subseteq \mathcal{M}_i$

        \State \textbf{return} $G$

      \EndProcedure
    \end{algorithmic}
  \end{algorithm}

  Given an assignment $y \in \{0,1\}^{rn}$ to all variables, let
  $y_j$ denote the restriction of $y$ to variables in the group
  $Y_j$.  Further let $I_{y_j}$ be the unique minterm of $\CM_j$
  associated with $y_j$.  From
  Proposition~\ref{prop:minterm-maxterm-ACM}~$(i)$, we know that
  $I_y = \bigtimes_{j \in [n]}I_{y_j}$ is a minterm of $\ACM$.
  Thus, we associate this minterm $I_y$ with the assignment $y$.
  We will now prove that $F$ is falsified by $y$ if and only if
  $F'$ has a hazard at $I_y$.

  To prove this, we consider the formula $F'$ in the algorithm
  just before collapsing $\AND$ gates of depths three and four
  (i.e., it has depth five).  The gates in $F'$ correspond to
  gates in $F$ in the following fashion: The output gate and
  gates of depth four in $F'$ correspond to the output gate and
  depth one gates in $F$ respectively and the gates of depth
  three in $F'$ correspond to literals in $F$.  Since all
  occurrences of literals in $F$ are replaced with formulas
  computing $\ACM$ in $F'$, the function computed at the output
  gate and all gates of depth four and three in $F'$ is also
  $\ACM$.

  Consider a gate $g$ at \emph{depth three} in $F'$. By
  construction, the sub-formula rooted at $g$ satisfies the
  property that it has a hazard at $I_y$ if and only if the
  corresponding literal in $F$ evaluates to $0$ on $y$.

  We now consider a gate $g$ at \emph{depth four} in $F'$.
  It is an $\AND$ gate. Assume that it evaluates to $0$ on input $y$ in $F$.
  Then, at least one of its inputs in $F$ must also evaluate to $0$ on $y$. 
  From the above argument about depth three gates, we know that the
  corresponding gate in $F'$ must have a hazard at the minterm $I_y$.
  Therefore, this gate must evaluate to $u$ on $I_y$ while
  the other inputs to the gate $g$ evaluate to $1$ or $u$.
  This is because we are evaluating an implementation of
  $\ACM$ on one of its minterms. Thus, the sub-formula rooted at
  $g$ in $F'$ must have a hazard at the minterm $I_y$ corresponding to $y$.
  In the other direction, suppose $g$ has a hazard at the
  minterm $I_y$. Then, at least one of its inputs must have a
  hazard at this minterm, which in turn implies that the corresponding
  literal in $F$ evaluates to $0$ on the assignment $y$.
  Since $g$ is an $\AND$ gate we thus obtain that $g$ evaluates to $0$ on $y$
  in $F$.   

  Finally we consider the $\OR$ gate $g$ at the root of $F'$.
  If $g$ outputs $u$ on $I_y$, all gates feeding into $g$ must
  output $u$ on $I_y$. (They cannot evaluate to $0$ because each of them
  is evaluating $\ACM$ function on a minterm.) 
  Therefore, all the corresponding gates in $F$ must output $0$ on
  $y$ causing $F$ to output $0$. On the other hand, if $F$
  outputs $0$ on $y$, every gate feeding into the root in $F$ must
  output $0$ on $y$ and therefore, all the corresponding gates in
  $F'$ must output $u$ on $I_y$ causing $F'$ to output $u$ as
  well.  \qed 
\end{proof}

\section{Detecting hazards in depth-two formulas}

We now look at the time complexity of detecting hazards in depth
two formulas.  We will focus on input formulas in \DNF. The dual
statements are true for formulas in \CNF. It is known that a
\DNF\ formula that does not contain terms with contradictory
literals (i.e., $x$ and $\NOT x$ for some variable $x$) cannot
have $0$-hazards.  Since it is trivial to remove such terms, the
interesting case is to detect 1-hazards in \DNF\ formulas.
M\'at\'e, Das, and Chuang \cite{MATE} gave an exponential time
algorithm that takes as input a \DNF\ formula and outputs an
equivalent \DNF\ formula that is hazard-free. Such an algorithm
is necessarily exponential time because there are functions that
have size $s$ \DNF\ formulas such that any hazard-free \DNF\
formulas for it has size at least $3^{s/3}$
\cite[Theorem~1.3]{CM78}. We show that there is a polynomial time
algorithm if the goal is to only \emph{detect} whether an input
\DNF\ formula has a 1-hazard.

We start with a crucial observation that help witness $1$-hazards
in \DNF\ easily. The simplest \DNF\ formula with a $1$-hazard is
$x \OR \NOT x$. We show that every \DNF\ formula with a
$1$-hazard has a $1$-hazard $\alpha \in \{0,1,u\}^n$ such that
when the \DNF\ is restricted by the stable values in $\alpha$,
the \emph{simplified} \DNF\ has the form \[x \OR \NOT x \OR H,\]
for some \DNF\ formula $H$ such that no term in $H$ evaluates to
$1$.  Obviously, it is easy to detect that there is a hazard in
such a simplified \DNF\ formulas.  We now introduce a definition
that will help us formally state the lemma.

\begin{defn}
  A \DNF\ $H$ over variables and constants is said to be \emph{equal} to $1$ if
  at least one of the terms in it evaluates to $1$. 
\end{defn}
For example, $x \OR \NOT x$ is \emph{not equal} to $1$, though it evaluates to $1$ on all possible inputs. On the other hand $x \OR \NOT x \OR 1$ is \emph{equal} to $1$, since it contains the term $1$ that trivially evaluates to $1$.  

\begin{lem}
  \label{lem:easy-hazards}
  Let $F$ be a \DNF\ on $n$ variables. Suppose that $F$ has a
  $1$-hazard. Then, there exists $\alpha \in \{0,1,u\}^n$ such
  that $F$ has a $1$-hazard at $\alpha$, and furthermore,
  \[F|_{\alpha} = x \OR \NOT x \OR H,\] for some variable $x$ and
  \DNF\ $H$ that is not equal to $1$.  Here, $F|_{\alpha}$
  represents the \DNF\ obtained by simplifying the terms of $F$
  upon substitution of variables by the stable values of $\alpha$. 
  A \DNF\ is simplified by exhaustively applying the following two rules:
  \begin{enumerate}
  \item[(i)] Remove terms with a literal that evaluates to $0$, 
  \item[(ii)] Shorten terms by the removal of literals that evaluate to $1$. 
  \end{enumerate} 
\end{lem}
\begin{proof}
  Let $\beta \in \{0,1,u\}^n$ be an arbitrary 1-hazard for $F$.
  Substitute the stable variables given by $\beta$ in $F$ to
  obtain a simplified \DNF\ $G$. If $G =  x \OR \NOT x \OR H$,
  for some variable $x$ and \DNF\ $H$, then $H$ is not
  equal to $1$ because $G$ must evaluate to $u$ on $\beta$ and, hence, 
  $\beta$ is the required $1$-hazard.  Suppose not, then either
  there exists a term of size $1$ in $G$ or every term is of size
  at least $2$.  In both cases we construct the required hazard
  from $\beta$ iteratively.  We will increase the number of
  stable values in $\beta$ at each step while ensuring that it
  remains a $1$-hazard.  In particular, we argue in both cases
  that there exists a variable $x$ in $G$ such that we can set it
  to $0$ and the resulting partial assignment is still a
  $1$-hazard.  Clearly this process terminates in at most $n$
  steps.  We now show how to find the variable in each case.

  Suppose there exists a term of size $1$ in $G$. That is,
  $G =\ell \OR H$, for some literal $\ell$ and, moreover, $H$ does
  not have $\NOT{\ell}$ as a term. Then we extend the partial
  assignment $\beta$ by setting $\ell = 0$.  We now claim that
  the new partial assignment is still a $1$-hazard.  This is
  easily seen because $\ell =0$ either kills a term in $G$ or
  reduces its size. It never makes a term evaluate to $1$, and
  therefore the hazard propagates.

  In the remaining case, every term in $G$ is of size at least
  $2$.  We pick an arbitrary literal from an arbitrary term and
  set it to $0$.  Again as before we can argue that the hazard
  propagates since any term is either killed or reduced in size,
  but never evaluated to $1$.

  Note that if $G$ has only one variable, then it must be $x \OR
  \NOT x$ for some $x$. This completes the proof of the lemma.
  \qed
\end{proof}
We now give a polynomial time algorithm to detect $1$-hazards in
\DNF\ formulas.
\begin{thm}
  \label{thm:poly-time-detection}
  There is a polynomial time algorithm that detects $1$-hazards
  in \DNF\ formulas.
\end{thm}
\begin{proof}
  Let $F$ be the input \DNF\ formula. From
  Lemma~\ref{lem:easy-hazards}, we know that to check whether $F$
  has a $1$-hazard it suffices to check for an
  \emph{easy-to-detect} hazard.  Observe that an easy-to-detect
  hazard is nothing but a partial assignment $\alpha$ such that
  $F|_\alpha$ has both $x$ and $\NOT x$ as terms for some
  variable $x$ and, furthermore, no term evaluates to $1$.
  In fact, we give a polynomial time procedure to
  find an easy-to-detect hazard. 

  To find such a partial assignment we do the following: For
  every pair of terms $S\AND x$ and $T\AND \NOT{x}$ in $F$, for
  some variable $x$, we check if $S$ and $T$ can be
  simultaneously set to $1$ while no other term in $F$ evaluates
  to $1$. If so, then clearly this partial assignment is an
  easy-to-detect hazard.

  It is easily seen that the above procedure runs in polynomial
  time. \qed
\end{proof}

Even though all $0$-hazards can be eliminated from \DNF\ formulas
by removing all terms with contradictory literals, the presence
of such terms do not imply a $0$-hazard.  For example, the single
variable \DNF\ formula $(x\AND\NOT{x}) \OR x$ does not have any
hazards.

In contrast to the poly-time algorithm to detect $1$-hazards, the
following simple reduction shows that detecting $0$-hazards in
\DNF\ formulas is hard.

\begin{thm}
  If \SETH\ is true, then there is no
  $O(2^{(1-\epsilon)n}\poly(s))$ time algorithm, for any
  $\epsilon > 0$, that detects 0-hazards in \DNF\ formulas on $n$
  variables and $s$ terms.
\end{thm}
\begin{proof}
  We will reduce the \DNF\ falsifiability problem to this
  problem.  Let $F$ be the input \DNF\ formula for the
  falsifiability problem.  We assume without loss of generality
  that $F$ does not contain terms with contradictory literals.
  We now claim that the \DNF\ formula $G = F \OR (x\AND\NOT{x})$,
  where $x$ is a new variable, has a $0$-hazard if and only if
  $F$ is falsifiable.  If $F$ is falsifiable at an input $a$,
  then the input $(a, u)$ is a $0$-hazard for $G$.  On the other
  hand if $F$ is a tautology, then so is $G$ and, therefore, $G$
  cannot have a $0$-hazard.\qed
\end{proof}

The above results can also be stated for \CNF\ formulas using the
following observation.

\begin{obs}
  A \CNF\ formula $F$ has a $0$-hazard if and only if the \DNF\
  formula $G = \neg F$ has a $1$-hazard.
\end{obs}
\begin{proof}
  Assume $F$ has a $0$-hazard at $a\in \{u, 0, 1\}^n$. Each
  clause in $F$ evaluates to $1$ or $u$ on input $a$. Then, the
  corresponding term in $G$ evaluates to $0$ or $u$ respectively
  which implies that the output of $G$ is also $u$. The other
  direction is similar.
\end{proof}

Observe that any \CNF\ formula that has a $1$-hazard contains a
clause that contains a variable and its negation. Therefore, we
can easily eliminate all $1$-hazards from a given $\CNF$ formula.
Combining these observations with theorems for $\DNF$ formulas,
we have the following theorem.

\begin{thm}
  There is a polynomial-time algorithm that detects $0$-hazards
  in \CNF\ formulas. Also, if \SETH\ is true, there is no
  $O(2^{(1-\epsilon)n}\poly(s))$ time algorithm for any $\epsilon
  > 0$ that detects $1$-hazards in \CNF\ formulas on $n$
  variables and $s$ clauses.
\end{thm}

\section{Conclusion}

We show that under \SETH\ the straightforward hazard detection
algorithm cannot be significantly improved upon, even when the
inputs are restricted to be depth-$4$ formulas. We also show that
there are polynomial time algorithms for detecting $1$-hazards
in \DNF\ formulas (resp., $0$-hazards in \CNF), while $0$-hazards
(resp., $1$-hazards) can be easily eliminated.  The complexity of
hazard detection for depth-$3$ formulas remain open. 

\section*{Acknowledgements}

The authors would like to thank the anonymous reviewers. Their
suggestions helped to greatly improve the presentation of results
in the paper.

\bibliographystyle{elsarticle-num}
\bibliography{seth}

\end{document}